\documentclass{amsart}
\usepackage{amsfonts}
\usepackage{blkarray}
\usepackage{float}
\usepackage{multirow}
\newtheorem{theorem}{Theorem}[section]
\newtheorem{lemma}[theorem]{Lemma}

\newtheorem{corollary}[theorem]{Corollary}
\theoremstyle{definition}

\newtheorem{example}[theorem]{Example}

\theoremstyle{remark}
\newtheorem{remark}[theorem]{Remark}
\numberwithin{equation}{section}

\begin{document}
\title[A modified four circulant construction for self-dual codes]{New extremal binary self-dual codes from a
modified four circulant construction}
\author{Abidin Kaya}
\author{Bahattin Yildiz}
\author{Abdullah Pasa}
\address{Department of Mathematics, Fatih University, 34500 Istanbul, Turkey}
\email{nabidin@gmail.com,byildiz@fatih.edu.tr,abdullah.pasa@fatih.edu.tr}
\subjclass[2010]{Primary 94B05; Secondary 94B99}
\keywords{extremal self-dual codes, circulant matrices, reverse-circulant matrices, four circulant construction}

\begin{abstract}
In this work, we propose a modified four circulant construction for
self-dual codes and a bordered version of the construction using the
properties of $\lambda$-circulant and $\lambda$-reverse circulant matrices.
By using the constructions on $\mathbb{F}_{2}$, we obtain new binary codes
of lengths 64 and 68. We also apply the constructions to the ring $R_2$ and
considering the $\mathbb{F}_{2}$ and $R_{1}$-extensions, we obtain new
singly-even extremal binary self-dual codes of lengths 66 and 68. More
precisely, we find 3 new codes of length 64, 15 new codes of length 66 and
22 new codes of length 68. These codes all have weight enumerators with
parameters that were not known to exist in the literature.
\end{abstract}

\maketitle

\section{Introduction}

Self-dual codes, especially over the binary field, have been a central topic
of research in coding theory for a considerable time now. This interest is
justified by the many combinatorial and algebraic objects that are related
to self-dual codes. The recent studies on self-dual codes are intensified
around classifying all extremal self-dual codes of given lengths. The
possible weight enumerators of all extremal binary self-dual codes of
lengths up to 100 are given in \cite{conway} and \cite{dougherty1}. Many of
these possible weight enumerators have parameters in them that fall in
certain intervals. A lot of recent research on extremal self-dual codes has
gone towards finding extremal self-dual codes with new parameters. \cite%
{karadeniz64}, \cite{karadenizfc}, \cite{kaya}, \cite{kayayildiz}, \cite%
{jacodes}, \cite{tufekci}, \cite{yankov} are examples of such works.

There are numerous constructions for extremal self-dual codes. The common
theme is a computer search, usually using the magma computer algebra (\cite%
{magma}), over a reduced search field. There are different methods to reduce
the search field. Some of the most common such methods use a special type of
matrices called \textit{circulant matrices}. An $n\times n$ circulant matrix
is determined uniquely by its first row. So, searching through such matrices
over an alphabet $A$, requires searching over $|A|^n$ matrices instead of $%
|A|^{n^2}$ matrices. Circulant matrices also have an important algebraic
property in that they commute with respect to the matrix multiplication.
These properties of circulant matrices have been used quite effectively in
such popular constructions in literature as double circulant, bordered
double circulant and four circulant constructions.

Another method of reducing the search field is to consider the above
mentioned constructions over suitable rings that are endowed with
orthogonality-preserving Gray maps. By considering the lifts of good binary
self-dual codes over such rings, the search field can be reduced even
further, giving way to many extremal codes with new parameters. This idea
has been used quite effectively in the above-mentioned works.

In this work, we modify the four-circulant construction, using $\lambda$%
-circulant, reverse-circulant and $\lambda$-reverse-circulant matrices.
Applying these modified constructions to three well-known alphabets, $%
\mathbb{F}_2$, $R_1 = \mathbb{F}_2+u\mathbb{F}_2$ and $R_2 = \mathbb{F}_{2}+u%
\mathbb{F}_{2}+v\mathbb{F}_{2}+uv\mathbb{F}_{2}$, we are able to find many
extremal binary self-dual codes with new parameters, thus extending the
database of all known such codes in the literature.

The rest of the work is organized as follows. In section 2, we give the
preliminaries for our constructions. We recall the alphabets that we have
mentioned above and we give a thorough background on circulant and
reverse-circulant matrices. In section 3, we introduce a modified
four-circulant construction, using $\lambda$-circulant and $\lambda$%
-reverse-circulant matrices as well as a bordered four circulant
construction. In section 4, we apply the constructions to the alphabets $%
\mathbb{F}_2$, $R_1 = \mathbb{F}_2+u\mathbb{F}_2$ and $R_2 = \mathbb{F}_{2}+u%
\mathbb{F}_{2}+v\mathbb{F}_{2}+uv\mathbb{F}_{2}$, as a result of which we
obtain many new extremal binary self-dual codes of lengths 64, 66 and 68.
More explicitly, we construct 3 new codes of length 64, 15 new codes of
length 66 and 22 new codes of length 68, adding these to the literature of
known such codes. The results are tabulated for each length.

\section{Preliminaries\label{preliminaries}}

\subsection{Basics}

Let $R$ be a finite ring. A linear \emph{code} $C$ of length $n$ \ over $R$
is an $R$-submodule of $R^{n}$. The elements of $C$ are called \emph{%
codewords}.

Let $\langle u,v\rangle $ be inner product of two codewords $u$ and $v$ in $%
R^{n}$ which is defined as $\langle u,v\rangle =\sum_{i=1}^{n}u_{i}v_{i}$,
where the operations are done in $R^{n}$. The \emph{dual code} of a code $C$
is $C^{\perp }=\{v\in R^{n}\mid \langle u,v\rangle =0$ for all $v\in C\}$.
If $C\subseteq C^{\perp }$, $C$ is called \emph{self-orthogonal}$,$ and $C$
is\emph{\ self-dual} if $C=C^{\perp }$.

The main case of interest for us is the case when $R=\mathbb{F}_2$, in which
case we obtain the usual binary self-dual codes. Binary self-dual codes are
called Type II if the weights of all codewords are multiples of 4 and Type I
otherwise. Rains finalized the upper bound for the minimum distance $d$ of a
binary self-dual code of length $n$ in \cite{rains} as $d\leq $ $4\lfloor
\frac{n}{24}\rfloor +6$ if $n\equiv 22\pmod{24}$ and $d\leq $ $4\lfloor
\frac{n}{24}\rfloor +4$, otherwise. A self-dual binary code is called
\textit{extremal }if it meets the bound.

The ring family of $R_{k}$ is introduced in \cite{us} as follows:%
\begin{equation*}
R_{k}=\mathbb{F}_{2}[u_{1},u_{2},\ldots
,u_{k}]/(u_{i}^{2}=0,u_{i}u_{j}=u_{j}u_{i})
\end{equation*}%
Leaving the details of these rings to the mentioned work, we just give basic
properties of first two members of this family which are used in this work.

The ring $R_{1}$ that is $\mathbb{F}_{2}+u\mathbb{F}_{2}$ is the first
example of the ring family of $R_{k}$. This ring was introduced in \cite%
{Bachoc} for constructing lattices and it has been studied quite extensively
in the literature of coding theory. Second member of \ $R_{k}$ is the ring $%
R_{2}$ that is $\mathbb{F}_{2}+u\mathbb{F}_{2}+v\mathbb{F}_{2}+uv\mathbb{F}%
_{2}.$ The ring $R_{2}$ was first introduced by Yildiz and Karadeniz in \cite{YK}. Unlike the ring $R_{2},$ the ring $R_{1}$ is finite chain ring.
Moreover, the ring $R_{k}$ is not finite chain ring for $k\geq 2$.

A linear distance preserving map from $R_{1}^{n}$ to $\mathbb{F}_{2}^{2n}$
is described in \cite{Dougherty} in terms of vectors as:%
\begin{equation*}
\phi _{1}(\bar{a}+u\bar{b})=(\bar{b},\bar{a}+\bar{b}),
\end{equation*}%
where $\bar{a},\bar{b}\in \mathbb{F}_{2}^{2n}$.

In \cite{YK}, the Lee weight was defined as the Hamming weight of the image
under a Gray map which was extended from that in \cite{Dougherty}.
Specifically, the following map for $R_{2}$ is given in \cite{YK} as
follows:
\begin{equation*}
\phi _{2}(\bar{a}+u\bar{b}+v\bar{c}+uv\bar{d})=(\bar{d},\bar{c}+\bar{d},\bar{%
b}+\bar{d},\bar{a}+\bar{b}+\bar{c}+\bar{d}),
\end{equation*}%
where $\bar{a},\bar{b},\bar{c},\bar{d}\in \mathbb{F}_{2}^{n}$.

One of the important properties of the Gray map that was given for this
family of rings was its orthogonality-preserving property. This gives rise
to the following important lemma which was given as Theorem 4 in \cite%
{SelfdualRk}:

\begin{lemma}
Let $C$ be a self-dual code over $R_{k}$ of length $n$ then $\phi _{k}(C)$
is a binary self-dual code of length $2^{k}n$. If $C$ is Type II code then $%
\phi _{k}(C)$ is Type II and if $C$ is Type I then $\phi _{k}(C)$ is Type I.
\end{lemma}

\subsection{The Circulant and Reverse-Circulant Matrices}

Since our constructions rely heavily on circulant and reverse-circulant
matrices, we give some of their properties in this subsection. Most of what
is explained here can be found in abstract algebra sources. However, we
wanted to put together some of these in a complete form.

With $R$ a commutative ring with identity, let $\sigma$ be the permutation
on $R^n$ that corresponds to the right shift, i.e.
\begin{equation}
\sigma(a_1,a_2, \dots, a_n) = (a_n,a_1, \dots, a_{n-1}).
\end{equation}
A circulant matrix is a square matrix where each row is a right-circular
shift of the previous row. In other words, if $\overline{r}$ is the first
row, a typical circulant matrix is of the form
\begin{equation}
\left[
\begin{array}{c}
\overline{r} \\ \hline
\sigma(\overline{r}) \\ \hline
\sigma^2(\overline{r}) \\ \hline
\vdots \\ \hline
\sigma^{n-1}(\overline{r})%
\end{array}
\right].
\end{equation}
It is clear that, with $T$ denoting the permutation matrix corresponding to
the $n$-cycle $(123...n)$, a circulant matrix with first row $(a_1,a_2,
\dots, a_n)$ can be expressed as a polynomial in $T$ as:
\begin{equation*}
a_1I_n+a_2T+a_3T^2+ \cdots + a_nT^{n-1}.
\end{equation*}
Since $T$ satisfies $T^n=I_n$, this shows that circulant matrices commute.

A reverse-circulant matrix is a square matrix where each row is a
left-circular shift of the previous row. It is clear to see that if $%
\overline{r}$ is the first row, a reverse-circulant matrix is of the form
\begin{equation}
\left[
\begin{array}{c}
\overline{r} \\ \hline
\sigma^{-1}(\overline{r}) \\ \hline
\sigma^{-2}(\overline{r}) \\ \hline
\vdots \\ \hline
\sigma^{-(n-1)}(\overline{r})%
\end{array}
\right].
\end{equation}

The following lemma recaptures some of the basic properties of circulant and
reverse-circulant matrices that can easily be verified:

\begin{lemma}
\label{circrevcirc} \textbf{(i)} Circulant matrices commute under matrix
multiplication. Reverse-circulant matrices do not necessarily commute under
multiplication.

\textbf{(ii)} Reverse-circulant matrices are symmetric.

\textbf{(iii)} Suppose that $A$ is a circulant matrix and $D$ is the
backdiagonal matrix with $1$ on its backdiagonal and $0$ elsewhere. Then
both $AD$ and $DA$ are reverse circulant.
\end{lemma}

Since $D^2 = I_n$, $D$ is its own inverse. This means that from (iii) of the
previous lemma we get the following corollary:

\begin{corollary}
\label{cor} Let $A$ be a reverse-circulant matrix. Then there exists unique
circulant matrices $B$ and $C$ such that $A = BD = DC$.
\end{corollary}

We prove the following important lemma using the properties above:

\begin{lemma}
\label{implem} \textbf{(i)} Circulant matrices are closed under matrix
multiplication. The product of two reverse-circulant matrices is circulant.

\textbf{(ii)} If $A$ is a circulant matrix and $B$ is a reverse-circulant
matrix, then $AB$ and $BA$ are both reverse-circulant matrices that are not
necessarily the same matrices.
\end{lemma}

\begin{proof}
\textbf{(i)} If $A_1$ and $A_2$ are circulant matrices then they correspond
to two polynomials in $T$ with $T^n=I_n$. It is clear that their product is
also a polynomial in $T$ of degree at most $n-1$.

Now suppose that $B_1$ and $B_2$ are two reverse-circulant matrices. By
Corollary \ref{cor}, we see that there exists circulant matrices $A_1$ and $%
A_2$ such that $B_1 = A_1D$ and $B_2=DA_2$. Since $D^2=I_n$, we get
\begin{equation*}
B_1B_2 = (A_1D)(DA_2) = A_1(D^2)A_2 = A_1A_2,
\end{equation*}
which is circulant by the first part.

\textbf{(ii)} Let $C_1$ and $C_2$ be circulant matrices such that $%
B=DC_1=C_2D$ by Corollary \ref{cor}. Then $AB = A(C_2D)= (AC_2)D$, which is
reverse-circulant by Lemma \ref{circrevcirc} (iii) and the first part.
Similarly, $BA = (DC_1)A = D(C_1A)$ is reverse-circulant.
\end{proof}

Let $\lambda$ be a unit in $R$. The circulant and reverse-circulant matrices
can be extended to the so-called $\lambda$-circulant and $\lambda$%
-reverse-circulant matrices. By $\sigma_{\lambda}$, we mean the map that
acts on $R^n$ as
\begin{equation}
\sigma_{\lambda}(a_1,a_2, \dots, a_n) = (\lambda a_n, a_1, \dots, a_{n-1}).
\end{equation}
A $\lambda$-circulant matrix is then a square matrix where each row is
obtained from the previous row by applying the map $\sigma_{\lambda}$. In
other words, a typical $\lambda$-circulant matrix is of the form
\begin{equation}
\left[
\begin{array}{c}
\overline{r} \\ \hline
\sigma_{\lambda}(\overline{r}) \\ \hline
\sigma_{\lambda}^2(\overline{r}) \\ \hline
\vdots \\ \hline
\sigma_{\lambda}^{n-1}(\overline{r})%
\end{array}
\right].
\end{equation}
Let $T_{\lambda}$ be the matrix that is obtained from the permutation matrix
$T$ by multiplying the $(n,1)$-entry by $\lambda$. Then a $\lambda$%
-circulant matrix with the first row $(a_1,a_2, \dots, a_n)$ can be
expressed as
\begin{equation*}
a_1I_n+a_2T_{\lambda}+a_3T_{\lambda}^2+ \cdots + a_nT_{\lambda}^{n-1}
\end{equation*}
with $T_{\lambda}^n = \lambda I_n.$ This shows that $\lambda$-circulant
matrices are closed under multiplication and that they commute with each
other as well.

By $\rho_{\lambda}$, let us define the analogous left $\lambda$-shift on $%
R^n $, namely:
\begin{equation}
\rho_{\lambda}(a_1,a_2, \dots, a_n) = (a_2, a_3, \dots, a_n, \lambda a_1).
\end{equation}
From the definitions, we can easily observe the following identities between
$\sigma_{\lambda}$ and $\rho_{\lambda}$
\begin{equation}
\sigma_{\lambda^{-1}} \circ \rho_{\lambda} = \rho_{\lambda} \circ
\sigma_{\lambda^{-1}} = 1.
\end{equation}

A $\lambda$-reverse-circulant matrix is defined as a square matrix where
each row is obtained from the previous row by applying the map $%
\rho_{\lambda}$. In other words, a typical $\lambda$-reverse-circulant
matrix is of the form
\begin{equation}
\left[
\begin{array}{c}
\overline{r} \\ \hline
\rho_{\lambda}(\overline{r}) \\ \hline
\rho_{\lambda}^2(\overline{r}) \\ \hline
\vdots \\ \hline
\rho_{\lambda}^{n-1}(\overline{r})%
\end{array}
\right].
\end{equation}

Similar to the properties of circulant and reverse-circulant matrices we
have the following results for $\lambda$-circulant and $\lambda$%
-reverse-circulant matrices.

\begin{lemma}
\label{lm} \textbf{(i)} The product of two $\lambda$-circulant matrices is
again $\lambda$-circulant while the product of a $\lambda$-circulant and $%
\gamma$-circulant matrix is not necessarily circulant if $\gamma \neq
\lambda $.

\textbf{(ii)} $\lambda$-circulant matrices commute under multiplication for
the same $\lambda$.

\textbf{(iii)} $\lambda $-reverse-circulant matrices are symmetric.

\textbf{(iv)} With $D$ denoting the backdiagonal matrix with $1$'s on its
backdiagonal, and $A$ denoting a $\lambda$-circulant matrix, we have: $AD$
is a $\lambda$-reverse-circulant matrix while $DA$ is a $\lambda^{-1}$%
-reverse-circulant matrix.

\textbf{(v)} For any $\lambda$-reverse-circulant matrix $A$, there exists a
unique $\lambda$-circulant matrix $B$ and a unique $\lambda^{-1}$-circulant
matrix $C$ such that $A = BD = DC$.
\end{lemma}

We prove the following result, which is analogous to Lemma \ref{implem}:

\begin{theorem}
\label{thm} Suppose $A$ is a $\lambda $-circulant matrix and $B$ be a $%
\lambda $-reverse-circulant matrix. Then $AB$ is a $\lambda $%
-reverse-circulant matrix for all units $\lambda \in R$. If $\lambda ^{2}=1$%
, then $BA$ is also $\lambda $-reverse-circulant.
\end{theorem}

\begin{proof}
By Lemma \ref{lm}-(v), let $C_1$ be the $\lambda$-circulant matrix such that
$B = C_1D$. Then $AB = A(C_1D) = (AC_1)D$. Since both $A$ and $C_1$ are $%
\lambda$-circulant, $AC_1$ is $\lambda$-circulant, and so by Lemma \ref{lm}%
-(iv), $(AC_1)D$ is $\lambda$-reverse-circulant.

For the second part, assume that $\lambda ^{2}=1$, $\lambda ^{-1}=\lambda $.
Again by the previous lemma, let $C_{2}$ be the $\lambda ^{-1}=\lambda $%
-circulant matrix such that $B=DC_{2}$. Then $BA=(DC_{2})A=D(C_{2}A)$.
Similar to the previous case, $C_{2}A$ is $\lambda $-circulant, and so by
the same lemma $BA=D(C_{2}A)$ is $\lambda ^{-1}=\lambda $-reverse-circulant.
\end{proof}

\section{Constructions\label{constructions}}

The four circulant construction were introduced in \cite{betsumiya} as;

Let $A$ and $B$ be $n\times n$ circulant matrices over $\mathbb{F}_{p}$ such
that $AA^{T}+BB^{T}=-I_{n}$ then the matrix%
\begin{equation*}
G=\left[ I_{2n\ }%
\begin{array}{|cc}
A & B \\
-B^{T} & A^{T}%
\end{array}%
\right]
\end{equation*}%
generates a self-dual code over $\mathbb{F}_{p}$.

Recently, the construction is applied on $\mathbb{F}_{2}+u\mathbb{F}_{2}$ in
\cite{karadenizfc}, which yielded to new binary self-dual codes.

In the following we give a modified version of the construction, which works
for any commutative Frobenius ring;

\begin{theorem}
\label{twoblock}Let $\lambda $ be a unit of the commutative Frobenius ring $%
R $, $A$ be a $\lambda $-circulant matrix and $B$ be a $\lambda $%
-reverse-circulant matrix with $AA^{T}+BB^{T}=-I_{n}$ then the matrix
\begin{equation*}
G=\left[ I_{2n\ }%
\begin{array}{|cc}
A & B \\
-B & A%
\end{array}%
\right]
\end{equation*}%
generates a self-dual code $\mathcal{C}$ over $R$.
\end{theorem}

\begin{proof}
Let $X=\left[
\begin{array}{cc}
A & B \\
-B & A%
\end{array}%
\right] $, it is enough to show that $XX^{T}=-I_{2n}$. We have%
\begin{equation*}
XX^{T}=\left[
\begin{array}{cc}
-I_{n} & -AB^{T}+BA^{T} \\
-BA^{T}+AB^{T} & -I_{n}%
\end{array}%
\right] .
\end{equation*}%
So, we need to verify that $-AB^{T}+BA^{T}=0$. By Lemma \ref{lm}-(iii) $B$
is a symmetric matrix so $AB^{T}=AB$ and by Theorem \ref{thm} $AB$ is a $%
\lambda $-reverse-circulant matrix and it is also symmetric by Lemma \ref{lm}%
-(iii). It follows that $AB^{T}=BA^{T}$, which implies $GG^{T}=0$. Since the
ring $R$ is Frobenius, the code $\mathcal{C}$ is self-dual.
\end{proof}

By applying Theorem \ref{twoblock} on $\mathbb{F}_{2}$ we obtain the codes
listed in Table \ref{tab:64binary}. Some examples of codes over $R_{2}$ are
given in Table \ref{tab:R2}.

For the rest of the manuscript let $\boldsymbol{J}_{n}$ denote the $n\times
n $ matrix with all entries equal to $1$.

\begin{lemma}
\label{rowsum}Let $A$ be a $\lambda $-circulant matrix over $R_{k}$, $B$ be
a reverse-circulant matrix and $r_{A}=\left( a_{1},a_{2},\ldots
,a_{n}\right) $ and $r_{B}=\left( b_{1},b_{2},\ldots ,b_{n}\right) $ denote
the first rows of $A$ and $B$, respectively. Let $S_{r_{A}}$ and $S_{r_{B}}$
denote the sum of the components of $r_{A}$ and $r_{B}$, respectively, i.e.$%
S_{r_{A}}=\sum_{i=1}^{n}a_{i}$. Then, $AA^{T}+BB^{T}=I_{n}+\boldsymbol{J}%
_{n} $ implies $S_{r_{A}}$ and $S_{r_{B}}$ are both units or both non-units
in $R_{k}$.
\end{lemma}

\begin{proof}
$AA^{T}+BB^{T}=I_{n}+\boldsymbol{J}_{n}$ implies $\left\langle
r_{A},r_{A}\right\rangle +\left\langle r_{B},r_{B}\right\rangle =0$, i.e. $%
\sum_{i=1}^{n}a_{i}^{2}+\sum_{i=1}^{n}b_{i}^{2}=0$. Then, $%
\sum_{i=1}^{n}a_{i}^{2}=\left( S_{r_{A}}\right) ^{2}$ and $%
\sum_{i=1}^{n}b_{i}^{2}=\left( S_{r_{B}}\right) ^{2}$ since $char\left(
R_{k}\right) =2$. Hence, $\left( S_{r_{A}}\right) ^{2}=\left(
S_{r_{B}}\right) ^{2}$ and the result follows from the fact (unit)$^{2}=1$,
(non-unit)$^{2}=0$ for elements of $R_{k}$.
\end{proof}

The following Theorem is a bordered version of the construction in Theorem %
\ref{twoblock}, which generates self-dual codes over the family of rings $%
R_{k}$.

\begin{theorem}
\label{bordered}Let $n$ be an odd number, $A$ be a circulant matrix or order
$n$ and $B$ be a reverse-circulant matrix of order $n$ with $%
S_{r_{A}}=S_{r_{B}}$, which is a unit and $AA^{T}+BB^{T}=I_{n}+\boldsymbol{J}%
_{n}$. Let $x$ be a unit and $y$ be a non-unit in $R_{k}$. Let $z=xS_{r_{A}}$
and $t=yS_{r_{B}}$ then the matrix
\begin{equation*}
G=\left[ I_{2n+2\ }%
\begin{array}{|cccc}
1 & 1 & \boldsymbol{x} & \boldsymbol{y} \\
1 & 1 & \boldsymbol{y} & \boldsymbol{x} \\
\boldsymbol{z}^{T} & \boldsymbol{t}^{T} & A & B \\
\boldsymbol{t}^{T} & \boldsymbol{z}^{T} & B & A%
\end{array}%
\right]
\end{equation*}%
generates a self-dual code $\mathcal{C}$ over $R_{k}$ where $\boldsymbol{x}%
=\left( x,x,\ldots ,x\right) $ denotes the all-$x$ vector, similarly for $%
\boldsymbol{y}$, $\boldsymbol{z}$ and $\boldsymbol{t}$.
\end{theorem}

\begin{proof}
Let $g_{i}$ denote the $i$-th row of $G$. $\left\langle
g_{1},g_{1}\right\rangle =1+n\left( x^{2}+y^{2}\right) =0$ since $n$ is odd,
$x^{2}=1$ and $y^{2}=0$. It is obvious that $\left\langle
g_{1},g_{2}\right\rangle =0$. By the choice of $z$ and $\ t$ we have $%
\left\langle g_{1},g_{i}\right\rangle =0$ for $3\leq i\leq 2n+2$. Similarly $%
g_{2}$ is orthogonal to itself and the other rows. Since $z$ is a unit and $%
t $ is a non-unit we have $z^{2}=1$ and $t^{2}=0$ which implies $%
\left\langle g_{i},g_{i}\right\rangle =1+z^{2}+t^{2}=0$ for $3\leq i\leq
2n+2.$ $AA^{T}+BB^{T}=I_{n}+\boldsymbol{J}_{n}$ and $z^{2}+t^{2}=1$ implies $%
\left\langle g_{i},g_{j}\right\rangle =0$ whenever $3\leq i<j\leq n+2$ or $%
n+3\leq i<j\leq 2n+2$. By Lemma \ref{implem}-(ii) $AB$ is a
reverse-circulant matrix and by Lemma \ref{circrevcirc}-(ii)
reverse-circulant matrices are symmetric which implies $B=B^{T}$ and
therefore $AB^{T}+BA^{T}=0$. Together with $zt+tz=0$ this implies $%
\left\langle g_{i},g_{j}\right\rangle =0$ whenever $3\leq i\leq n+2$ or $%
n+3\leq j\leq 2n+2$. Hence, $\mathcal{C}$ is a self-dual code over $R_{k}$.
\end{proof}

\begin{example}
Let $\mathcal{C}$ be the code over $R_{1}$ generated by
\begin{equation*}
\left[ I_{16\ }%
\begin{array}{|cccc}
1 & 1 & \boldsymbol{1} & \boldsymbol{u} \\
1 & 1 & \boldsymbol{u} & \boldsymbol{1} \\
\left( \boldsymbol{1+u}\right) ^{T} & \boldsymbol{u}^{T} & A & B \\
\boldsymbol{u}^{T} & \left( \boldsymbol{1+u}\right) ^{T} & B & A%
\end{array}%
\right]
\end{equation*}%
where $r_{A}=\left( u,0,1,1,u,1,u\right) $ and $r_{B}=\left(
0,0,0,1,u,u,u\right) $ then $S_{r_{A}}=S_{r_{B}}=1+u$ and $%
AA^{T}+BB^{T}=I_{7}+\boldsymbol{J}_{7}$. Thus, $\mathcal{C}$ is a self-dual
code by Theorem \ref{bordered}. The binary image $\phi _{1}\left( \mathcal{C}%
\right) $ is a self-dual Type II $\left[ 64,32,12\right] _{2}$-code.
\end{example}

The codes listed in Table \ref{tab:bordered} are obtained by applying
Theorem \ref{bordered} on $\mathbb{F}_{2}$.

\section{New binary self-dual codes of lengths 64, 66 and 68\label{newcodes}}

In this section, we apply the constructions given in Section \ref%
{constructions}. By using Theorem \ref{twoblock} on $\mathbb{F}_{2}$ we were
able to construct 4 new codes of length 68. Three new extremal binary
self-dual codes of length 64 are constructed as an application of Theorem %
\ref{bordered} on $\mathbb{F}_{2}$. Moreover, by considering the extensions
of the codes new codes of length 66 and 68 are constructed. The following
theorem is used for extensions:

\begin{theorem}
$($\cite{doughertyfrobenius}$)$ \label{ext}Let $\mathcal{C}$ be a self-dual
code over $R$ of length $n$ and $G=(r_{i})$ be a $k\times n$ generator
matrix for $\mathcal{C}$, where $r_{i}$ is the $i$-th row of $G$, $1\leq
i\leq k$. Let $c$ be a unit in $R$ such that $c^{2}=1$ and $X$ be a vector
in $R^{n}$ with $\left\langle X,X\right\rangle =1$. Let $y_{i}=\left\langle
r_{i},X\right\rangle $ for $1\leq i\leq k$. Then the following matrix%
\begin{equation*}
\left(
\begin{array}{cc|c}
1 & 0 & X \\ \hline
y_{1} & cy_{1} & r_{1} \\
\vdots & \vdots & \vdots \\
y_{k} & cy_{k} & r_{k}%
\end{array}%
\right) ,
\end{equation*}%
generates a self-dual code $\mathcal{C}^{\prime }$ over $R$ of length $n+2$.
\end{theorem}

\subsection{New extremal binary self-dual codes of length 64}

There are two possibilities for the weight enumerators of extremal
singly-even $\left[ 64,32,12\right] _{2}$ codes (\cite{conway}):
\begin{eqnarray*}
W_{64,1} &=&1+\left( 1312+16\beta \right) y^{12}+\left( 22016-64\beta
\right) y^{14}+\cdots ,\text{ }14\leq \beta \leq 284, \\
W_{64,2} &=&1+\left( 1312+16\beta \right) y^{12}+\left( 23040-64\beta
\right) y^{14}+\cdots ,\text{ }0\leq \beta \leq 277.
\end{eqnarray*}%
Together with the ones constructed in \cite{karadeniz64,karadenizfc,yankov},
codes exist with weight enumerators $\beta =$14, 18, 22, 25, 32, 36, 39, 44,
46, 53, 60, 64 in $W_{64,1}$ and for $\beta =$0, 1, 2, 4, 5,\ 6, 8, 9, 10,
12, 13,\ 14, 16,\ 17, 18, 20, 21,\ 22, 23, 24,\ 25,\ 28,\ 29,\ 30, 32,\ 33,\
36, 37, 38, 40,\ 41,\ 44, 48, 51,\ 52,\ 56, 58, 64, 72, 80,\ 88,\ 96, 104,
108,\ 112,\ 114,\ 118,\ 120 and 184 in $W_{64,2}$.

In this work, we construct the codes with weight enumerators $\beta =$29,
59, 74 in $W_{64,1}$.

By Theorem \ref{twoblock} we obtain extremal binary self-dual codes of
length 64 that are given in Table \ref{tab:64binary}.
\begin{table}[tbp]
\caption{Extremal binary self-dual Type I codes of length 64 by Theorem\
\protect\ref{twoblock} on $\mathbb{F}_{2}$. }
\label{tab:64binary}\centering
\begin{tabular}{||c|c|c|c||c||}
\hline
$\mathcal{B}_{64,i}$ & $r_{A}$ & $r_{B}$ & $\beta $ in $W_{64,2}$ & $%
\left\vert Aut\left( \mathcal{B}_{64,i}\right) \right\vert $ \\ \hline\hline
$\mathcal{B}_{64,1}$ & $\left( 0101110001100111\right) $ & $\left(
1011010101010100\right) $ & 0 & $2^{5}$ \\ \hline
$\mathcal{B}_{64,2}$ & $\left( 1010101011110101\right) $ & $\left(
0110110001001001\right) $ & 8 & $2^{5}$ \\ \hline
$\mathcal{B}_{64,3}$ & $\left( 0111000010111110\right) $ & $\left(
0011001111010010\right) $ & 16 & $2^{5}$ \\ \hline
$\mathcal{B}_{64,4}$ & $\left( 1011001001100101\right) $ & $\left(
0110100101000000\right) $ & 24 & $2^{5}$ \\ \hline
$\mathcal{B}_{64,5}$ & $\left( 1100110010100010\right) $ & $\left(
1110010010111110\right) $ & 32 & $2^{5}$ \\ \hline
$\mathcal{B}_{64,6}$ & $\left( 1000111110101000\right) $ & $\left(
0110100011011101\right) $ & 40 & $2^{5}$ \\ \hline
$\mathcal{B}_{64,7}$ & $\left( 0100010000111100\right) $ & $\left(
1011101010000001\right) $ & 48 & $2^{5}$ \\ \hline
$\mathcal{B}_{64,8}$ & $(1111011010000100)$ & $(1101010101100011)$ & 56 & $%
2^{5}$ \\ \hline
$\mathcal{B}_{64,9}$ & $\left( 1000110110110001\right) $ & $\left(
1011001001101011\right) $ & 64 & $2^{6}$ \\ \hline
$\mathcal{B}_{64,10}$ & $\left( 0101110111001111\right) $ & $\left(
0001000111000110\right) $ & 72 & $2^{5}$ \\ \hline
\end{tabular}%
\end{table}

In order to fit the upcoming tables we use hexadecimal number sytem. The
one-to-one correspondence between hexadecimals and binary $4$ tuples is as
follows:
\begin{eqnarray*}
0 &\leftrightarrow &0000,\ 1\leftrightarrow 0001,\ 2\leftrightarrow 0010,\
3\leftrightarrow 0011, \\
4 &\leftrightarrow &0100,\ 5\leftrightarrow 0101,\ 6\leftrightarrow 0110,\
7\leftrightarrow 0111, \\
8 &\leftrightarrow &1000,\ 9\leftrightarrow 1001,\ A\leftrightarrow 1010,\
B\leftrightarrow 1011, \\
C &\leftrightarrow &1100,\ D\leftrightarrow 1101,E\leftrightarrow 1110,\
F\leftrightarrow 1111.
\end{eqnarray*}

To express the elements of $R_{2}$ we use the ordered basis $\left\{
uv,v,u,1\right\} $. As an application of Theorem \ref{twoblock} on $R_{2}$
binary self-dual Type I codes of length 64 are constructed, which are listed
in Table \ref{tab:R2}.

\begin{table}[tbp]
\caption{Extremal self-dual $\left[ 64,32,12\right] _{2}$ codes as Gray
images by Theorem\ \protect\ref{twoblock} on $R_{2}$.}
\label{tab:R2}\centering
\begin{tabular}{||c|c|c|c|c||c||}
\hline
$\mathcal{D}_{64,i}$ & $\lambda $ & $r_{A}$ & $r_{B}$ & $\beta $ in $%
W_{64,2} $ & $\left\vert Aut\left( \mathcal{D}_{64,i}\right) \right\vert $
\\ \hline\hline
$\mathcal{D}_{64,1}$ & $5$ & $(6,5,A,E)$ & $(D,7,D,5)$ & 0 & $2^{5}$ \\
\hline
$\mathcal{D}_{64,2}$ & $B$ & $\left( 1,6,3,B\right) $ & $\left(
1,8,7,4\right) $ & 1 & $2^{3}$ \\ \hline
$\mathcal{D}_{64,3}$ & $7$ & $(9,7,9,D)$ & $(5,C,A,A)$ & 4 & $2^{3}$ \\
\hline
$\mathcal{D}_{64,4}$ & $7$ & $(5,2,3,9)$ & $(7,4,3,8)$ & 5 & $2^{3}$ \\
\hline
$\mathcal{D}_{64,5}$ & $7$ & $(6,F,A,B)$ & $(C,9,3,1)$ & 8 & $2^{3}$ \\
\hline
$\mathcal{D}_{64,6}$ & $D$ & $(9,E,D,B)$ & $(F,A,5,0)$ & 9 & $2^{3}$ \\
\hline
$\mathcal{D}_{64,7}$ & $7$ & $(4,9,F,9)$ & $(7,8,B,6)$ & 12 & $2^{4}$ \\
\hline
$\mathcal{D}_{64,8}$ & $D$ & $(F,E,3,3)$ & $(D,2,B,4)$ & 13 & $2^{3}$ \\
\hline
$\mathcal{D}_{64,9}$ & $3$ & $(3,3,5,B)$ & $(8,4,2,9)$ & 16 & $2^{5}$ \\
\hline
$\mathcal{D}_{64,10}$ & $B$ & $(D,4,7,9)$ & $(D,6,B,0)$ & 17 & $2^{3}$ \\
\hline
$\mathcal{D}_{64,11}$ & $D$ & $(3,7,9,B)$ & $(3,2,E,E)$ & 20 & $2^{3}$ \\
\hline
$\mathcal{D}_{64,12}$ & $7$ & $(E,7,0,4)$ & $(B,3,5,7)$ & 21 & $2^{3}$ \\
\hline
$\mathcal{D}_{64,13}$ & $7$ & $(C,1,1,9)$ & $(5,2,F,8)$ & 24 & $2^{4}$ \\
\hline
$\mathcal{D}_{64,14}$ & $7$ & $(8,B,6,2)$ & $(1,1,D,F)$ & 25 & $2^{3}$ \\
\hline
$\mathcal{D}_{64,15}$ & $B$ & $(C,D,0,3)$ & $(D,E,7,5)$ & 28 & $2^{5}$ \\
\hline
$\mathcal{D}_{64,16}$ & $B$ & $(6,7,C,4)$ & $(F,D,9,D)$ & 29 & $2^{3}$ \\
\hline
$\mathcal{D}_{64,17}$ & $3$ & $(5,4,1,4)$ & $(7,B,7,6)$ & 32 & $2^{5}$ \\
\hline
$\mathcal{D}_{64,18}$ & $D$ & $(8,5,4,2)$ & $(1,B,5,1)$ & 33 & $2^{3}$ \\
\hline
$\mathcal{D}_{64,19}$ & $B$ & $(9,9,C,3)$ & $(8,1,A,F)$ & 36 & $2^{3}$ \\
\hline
$\mathcal{D}_{64,20}$ & $5$ & $(E,A,D,6)$ & $(F,3,B,D)$ & 48 & $2^{5}$ \\
\hline
$\mathcal{D}_{64,21}$ & $D$ & $(6,9,0,3)$ & $(A,9,3,1)$ & 64 & $2^{5}$ \\
\hline
$\mathcal{D}_{64,22}$ & $5$ & $(A,9,D,1)$ & $(F,8,5,E)$ & 80 & $2^{7}$ \\
\hline
\end{tabular}%
\end{table}

\begin{remark}
Note that the first codes with weight enumerators $\beta =$1, 5, 13, 17, 21,
25, 29, 33 and 80 in $W_{64,2}$ are recently constructed in \cite%
{karadeniz64,karadenizfc}. Those are reconstructed in Table \ref{tab:R2}.
\end{remark}

The construction in Theorem \ref{bordered} yields extremal binary self-dual
codes of length 64. Those are listed in Table \ref{tab:bordered}, codes with
weight enumerators $\beta =$29, 59 and 74 in $W_{64,1}$ are constructed for
the first time.

\begin{table}[tbp]
\caption{Extremal self-dual Type I $\left[ 64,32,12\right] _{2}$ codes by
Theorem \protect\ref{bordered} on $\mathbb{F}_{2}$.}
\label{tab:bordered}\centering
\begin{tabular}{||c|c|c|c||c||}
\hline
$\mathcal{C}_{64,i}$ & $r_{A}$ & $r_{B}$ & $\beta $ in $W_{64,1}$ & $%
\left\vert Aut\left( \mathcal{C}_{64,i}\right) \right\vert $ \\ \hline\hline
$\mathcal{C}_{64,1}$ & $\left( 001101000000011\right) $ & $\left(
011000010011011\right) $ & 14 & $2^{2}\times 3\times 5$ \\ \hline
$\mathcal{C}_{64,2}$ & $\left( 010001101111110\right) $ & $\left(
111111100011110\right) $ & 14 & $2\times 3\times 5$ \\ \hline
$\mathcal{C}_{64,3}$ & $(001101111000010)$ & $(110010110110011)$ & \textbf{29%
} & $2\times 3\times 5$ \\ \hline
$\mathcal{C}_{64,4}$ & $\left( 111010001101101\right) $ & $\left(
100101000001111\right) $ & 44 & $2\times 3\times 5$ \\ \hline
$\mathcal{C}_{64,5}$ & $\left( 101000110101111\right) $ & $\left(
000000000011100\right) $ & 44 & $2^{2}\times 3\times 5$ \\ \hline
$\mathcal{C}_{64,6}$ & $(101101011101111)$ & $(001000001110001)$ & \textbf{59%
} & $2\times 3\times 5$ \\ \hline
$\mathcal{C}_{64,7}$ & $(011000100111111)$ & $(011000000000010)$ & \textbf{74%
} & $2^{2}3\times 5$ \\ \hline
\end{tabular}%
\end{table}

\subsection{New extremal binary self-dual codes of length 66}

A self-dual $\left[ 66,33,12\right] _{2}$-code has a weight enumerator in
one of the following forms (\cite{dougherty1}):%
\begin{eqnarray*}
W_{66,1} &=&1+\left( 858+8\beta \right) y^{12}+\left( 18678-24\beta \right)
y^{14}+\cdots ,\text{ where }0\leq \beta \leq 778, \\
W_{66,2} &=&1+1690y^{12}+7990y^{14}+\cdots \text{ } \\
\text{and }W_{66,3} &=&1+\left( 858+8\beta \right) y^{12}+\left(
18166-24\beta \right) y^{14}+\cdots ,\text{ where }14\leq \beta \leq 756,
\end{eqnarray*}%
For a list of known codes in $W_{66,1}$ we refer to \cite%
{karadenizfc,tufekci,yankovieee} and codes with weight enumerator $W_{66,2}$
exist. Together with the ones constructed in \cite{kaya} the existence of
the codes is known for $\beta =$28, 29, 30, 31, 32, 33, 34, 35, 36, 37, 38,
43, 44, 45, 46, 47, 48, 49, 50, 51, 54, 55, 56, 57, 58, 59, 60, 61, 62, 63,
66, 67,\ 70, 71, 73, 74, 75, 76, 77, 78, 79 and 80\ in $W_{66,3}$.

In this work, we construct 15 codes with new weight enumerators in $W_{66,3}$%
. More precisely, the codes with weight enumerators $\beta =$46, 52, 53, 61,
64, 81, 82, 83, 84, 85, 86, 87, 88, 90 and 92 in $W_{66,3}$ are constructed
for the first time in the literature. The codes that are obtained by
applying Theorem \ref{ext} are listed in Table \ref{tab:66table1} and Table %
\ref{tab:66table2}.
\begin{table}[tbp]
\caption{New codes in $W_{66,3}$ by Theorem \protect\ref{ext} $\mathbb{F}%
_{2} $. (10 codes)}
\label{tab:66table1}\centering
\begin{tabular}{||c|c||c||}
\hline
$\mathcal{C}_{i}$ & $X$ & $\beta $ in $W_{66,3}$ \\ \hline\hline
$\mathcal{C}_{64,5}$ & $\left( 11001101110010001101111011101100\boldsymbol{1}%
^{\boldsymbol{32}}\right) $ & \textbf{52} \\ \hline
$\mathcal{C}_{64,4}$ & $\left( 11110101011010111010110100000001\boldsymbol{1}%
^{\boldsymbol{32}}\right) $ & \textbf{61} \\ \hline
$\mathcal{C}_{64,5}$ & $\left( 00100101011000001000111010101100\boldsymbol{1}%
^{\boldsymbol{32}}\right) $ & \textbf{64} \\ \hline
$\mathcal{C}_{64,7}$ & $\left( 11110100100100011110100101100101\boldsymbol{1}%
^{\boldsymbol{32}}\right) $ & \textbf{81} \\ \hline
$\mathcal{C}_{64,7}$ & $\left( 11110011010001000000111101011110\boldsymbol{1}%
^{\boldsymbol{32}}\right) $ & \textbf{83} \\ \hline
$\mathcal{C}_{64,7}$ & $\left( 00101011111010100110001001011110\boldsymbol{1}%
^{\boldsymbol{32}}\right) $ & \textbf{84} \\ \hline
$\mathcal{C}_{64,7}$ & $(00101111000111010000101010111101\boldsymbol{1}^{%
\boldsymbol{32}})$ & \textbf{85} \\ \hline
$\mathcal{C}_{64,7}$ & $\left( 11001000011100000101010011000110\boldsymbol{1}%
^{\boldsymbol{32}}\right) $ & \textbf{87} \\ \hline
$\mathcal{C}_{64,7}$ & $\left( 11111000010000100011011010100101\boldsymbol{1}%
^{\boldsymbol{32}}\right) $ & \textbf{90} \\ \hline
$\mathcal{C}_{64,7}$ & $\left( 01010110000100110011000110000011\boldsymbol{1}%
^{\boldsymbol{32}}\right) $ & \textbf{92} \\ \hline
\end{tabular}%
\end{table}
\begin{table}[tbp]
\caption{New codes in $W_{66,3}$ by Theorem \protect\ref{ext} $\mathbb{F}%
_{2} $. (5 codes)}
\label{tab:66table2}\centering
\begin{tabular}{||c|c||c||}
\hline
$\mathcal{C}_{64,i}$ & $X$ & $\beta $ \\ \hline\hline
$\mathcal{C}_{64,3}$ & $\left(
1100110010100000010111010111000010110000110011010000111001101100\right) $ &
\textbf{46} \\ \hline
$\mathcal{C}_{64,5}$ & $\left(
0010001010110101110110100110011000110110101100100000110000111101\right) $ &
\textbf{53} \\ \hline
$\mathcal{C}_{64,7}$ & $\left(
0000101100110000000001100100101110100001010010101111110011001001\right) $ &
\textbf{82} \\ \hline
$\mathcal{C}_{64,7}$ & $\left(
0011000101011100001001101011011000101100110100101110000011100010\right) $ &
\textbf{86} \\ \hline
$\mathcal{C}_{64,7}$ & $\left(
1110110000111111101101111011001110101101010101100100001101111001\right) $ &
\textbf{88} \\ \hline
\end{tabular}%
\end{table}

\subsection{New extremal binary self-dual codes of length 68}

The weight enumerator of a self-dual $\left[ 68,34,12\right] _{2}$ code is
in one of the following forms (\cite{dougherty1}):
\begin{eqnarray*}
W_{68,1} &=&1+\left( 442+4\beta \right) y^{12}+\left( 10864-8\beta \right)
y^{14}+\cdots , \\
W_{68,2} &=&1+\left( 442+4\beta \right) y^{12}+\left( 14960-8\beta
-256\gamma \right) y^{14}+\cdots
\end{eqnarray*}%
where $\beta $ and $\gamma $ are parameters. Recently, $32$ new codes are
obtained in \cite{kaya} and $178$ new codes including the first examples
with $\gamma =3$ in $W_{68,2}$ are obtained in \cite{kayayildiz}. For a list
of known codes in $W_{68,1}$ we refer to \cite{yankovieee,tufekci}.
Recently, new codes in $W_{68,2}$ are obtained in \cite%
{kayayildiz,kaya,tufekci,yankovieee} together with these, codes exist for $%
W_{68,2}$ when
\begin{eqnarray*}
\gamma &=&0,\ \beta =11,22,33,44,\dots ,154,165,187,209,231,303 \\
\text{ or }\beta &\in &\left\{ 2m|m=\text{17, 20, 88, 99, 102, 110, 119,
136, 165 or }78\leq m\leq 86\right\} ; \\
\gamma &=&1,\ \beta =49,57,59,\dots ,160\text{ or }\beta \in \left\{ 2m|m=%
\text{27, 28, 29, 95, 96 or }81\leq m\leq 89\right\} ; \\
\gamma &=&2,\ \beta =65,69,71,77,81,159,186\text{ or }\beta \in \left\{
2m|30\leq m\leq 68,\text{ }70\leq m\leq 91\right\} \text{ or} \\
\beta &\in &\left\{ 2m+1|42\leq m\leq 69,\text{ }71\leq m\leq 77\right\} ; \\
\gamma &=&3,\ \beta =101,107,117,123,127,133,137,141,145,147,149,153,159,193%
\text{ or } \\
\beta &\in &\left\{ 2m|m=\text{44,45,47,48,50,51,52,54,}\dots \text{,61,63,}%
\dots \text{,66,68,}\dots \text{,72,74,77,}\dots \text{,84,86,}\dots \text{%
,90,94,98}\right\} ; \\
\gamma &=&4\text{, }\beta \in \left\{ 2m|m=\text{%
43,48,49,51,52,54,55,56,58,60,61,62,64,65,67,}\dots \text{,71,75,}\dots
\text{,78,80,87,97}\right\} ;\text{ } \\
\gamma &=&6\text{ with }\beta \in \left\{ 2m|m=\text{69, 77, 78, 79, 81, 88}%
\right\} \text{.}
\end{eqnarray*}

We construct the codes with weight enumerators in $W_{68,2}$ for $\gamma =0$
and $\beta =$17, 155, 157, 187, 221 and 255; $\gamma =3$ and $\beta =$103,
105, 115, 119, 121, 124, 125, 129, 131, 134, 150, 178, 182, 184, 190 and 194.

By using Theorem \ref{twoblock} we were able to obtain codes with weight
enumerators $\gamma =0$ and $\beta =$\textbf{17}, 34, 51, 68, 85, 102, 119,
136, 153, 170, \textbf{187}, 204, \textbf{221}, 238, \textbf{255}, 272 for $%
W_{68,2}$. In order to save space we only list the new codes in Table \ref%
{tab:68table1}.
\begin{table}[]
\caption{New codes in $W_{68,2}$ by Theorem \protect\ref{twoblock} on $%
\mathbb{F}_{2}$. (4 codes)}
\label{tab:68table1}\centering
\begin{tabular}{||c|c|c|c||c||}
\hline
$\mathcal{C}_{68,i}$ & $r_{A}$ & $r_{B}$ & $\beta $ in $W_{68,2}$ & $%
\left\vert Aut\left( \mathcal{C}_{68,i}\right) \right\vert $ \\ \hline
$\mathcal{C}_{68,1}$ & $\left( 01111110101111011\right) $ & $\left(
11001000101001011\right) $ & \textbf{17} & $2\times 17$ \\ \hline
$\mathcal{C}_{68,2}$ & $\left( 11110001011001010\right) $ & $\left(
11010100001011010\right) $ & \textbf{187} & $2\times 17$ \\ \hline
$\mathcal{C}_{68,3}$ & $(00110001101111011)$ & $(01000010000000100)$ &
\textbf{221} & $2\times 17$ \\ \hline
$\mathcal{C}_{68,4}$ & $\left( 11010010110010011\right) $ & $\left(
10100001001111100\right) $ & \textbf{255} & $2\times 17$ \\ \hline
\end{tabular}%
\end{table}

\begin{example}
Let $\mathcal{C}_{1}$ and $\mathcal{C}_{2}$ be the $R_{1}$-extensions of $%
\varphi _{u}\left( \mathcal{D}_{64,21}\right) ~$with respect to Theorem
respectively with $c_{1}=1+u,c_{2}=1$ and
\begin{eqnarray*}
X_{1} &=&\left(
3,u,0,0,0,0,1,u,3,0,3,u,1,1,0,0,u,1,1,0,1,3,1,u,1,3,0,u,0,0,3,3\right) \\
X_{2} &=&\left(
1,u,u,0,0,u,1,0,3,0,3,0,1,3,u,0,u,1,1,u,3,3,1,u,1,1,u,0,u,u,1,1\right) .
\end{eqnarray*}%
Then, $\varphi \left( \mathcal{C}_{1}\right) $ and $\varphi \left( \mathcal{C%
}_{2}\right) $ are extremal binary self-dual codes of length 68 with weight
enumerators respectively $\gamma =0,\beta =155$ and $\gamma =0,\beta =157$
in $W_{68,2}$.
\end{example}

The codes over $R_{2}$ in Table \ref{tab:R2} are mapped to $R_{1}$ and by
applying the extension theorem we were able to obtain 16 new codes of length
$68$ as binary images of the extensions. The codes are listed in Table \ref%
{tab:68table2}.

\begin{table}[]
\caption{New extremal binary self-dual codes of length 68 with $\protect%
\gamma =3$ in $W_{68,2}$ by Theorem \protect\ref{ext} on $R_{1}$(16 codes)}
\label{tab:68table2}\centering
\begin{tabular}{||c|c|c||c||}
\hline
Code & $c$ & $X$ & $\beta $ in $W_{68,2}$ \\ \hline\hline
$\mathcal{D}_{64,10}$ & $1$ & $\left(
1u1001030u3103111u3130u01u0u0331\right) $ & $\mathbf{103}$ \\ \hline
$\mathcal{D}_{64,10}$ & $3$ & $\left(
303u0101uu3301113u31300u1000u113\right) $ & $\mathbf{105}$ \\ \hline
$\mathcal{D}_{64,10}$ & $1$ & $\left(
3u1u03u3u03303331u313u0u30uuu331\right) $ & $\mathbf{115}$ \\ \hline
$\mathcal{D}_{64,10}$ & $1$ & $\left(
3010u1u10u1103313u111u0u3uu00113\right) $ & $\mathbf{119}$ \\ \hline
$\mathcal{D}_{64,10}$ & $3$ & $\left(
301001030u1303131u31100u3u000133\right) $ & $\mathbf{121}$ \\ \hline
$\mathcal{D}_{64,10}$ & $3$ & $\left(
uuu101u303u3u11uu3u1003uu1u1001u\right) $ & $\mathbf{124}$ \\ \hline
$\mathcal{D}_{64,10}$ & $1$ & $\left(
1u1uu3030u11u3133u3330uu3u00u333\right) $ & $\mathbf{125}$ \\ \hline
$\mathcal{D}_{64,10}$ & $3$ & $\left(
1u1003u10u33u313303330u01000u111\right) $ & $\mathbf{129}$ \\ \hline
$\mathcal{D}_{64,10}$ & $3$ & $\left(
1u10u1030u13u31330331u001uu0u111\right) $ & $\mathbf{131}$ \\ \hline
$\mathcal{D}_{64,10}$ & $3$ & $\left(
1u303011uu01000u10303133u0u10u33\right) $ & $\mathbf{134}$ \\ \hline
$\mathcal{D}_{64,10}$ & $1$ & $\left(
000101u10103u110u1030u30u101uu30\right) $ & $\mathbf{150}$ \\ \hline
$\mathcal{D}_{64,22}$ & $1$ & $\left(
3u000103031u00u30uu03u30u0uu11u1\right) $ & $\mathbf{178}$ \\ \hline
$\mathcal{D}_{64,22}$ & $1$ & $\left(
3u0uu3u3u13uuuu3uuuu1u3uuu0031u3\right) $ & $\mathbf{182}$ \\ \hline
$\mathcal{D}_{64,22}$ & $3$ & $\left(
u1100u001uu0uuu311331u101u03111u\right) $ & $\mathbf{184}$ \\ \hline
$\mathcal{D}_{64,22}$ & $1$ & $\left(
1uu00301u33uuu03u00u3u1u0u0031u3\right) $ & $\mathbf{190}$ \\ \hline
$\mathcal{D}_{64,22}$ & $1$ & $\left(
30000101013u00u100003u100u001303\right) $ & $\mathbf{194}$ \\ \hline
\end{tabular}%
\end{table}

\begin{theorem}
The existence of extremal self-dual binary codes is known for $15$
parameters in $W_{64,1};$ $57$ parameters in $W_{66,3}$ and $465$ parameters
in $W_{68,2}.$
\end{theorem}

\begin{remark}
The binary generator matrices of the\ constructed new codes are available
online at \cite{web}.
\end{remark}

\end{document}